\documentclass[smallextended]{svjour3}       
\smartqed  
\usepackage{graphicx}
\usepackage{amssymb, mathtools, enumerate}

\newcommand{\Nat}{{\mathbb N}}

\newcommand{\bi}{\chi_b}
\newcommand{\bhom}{\xrightarrow{b}}
\newcommand{\Gl}{G[K_\ell]}
\newcommand{\Gil}[1]{G_{#1}[K_\ell]}

\journalname{Graphs and Combinatorics}
\begin{document}

\title{On the b-continuity of the lexicographic product of graphs
\thanks{All authors are members of the ParGO Research Group - Parallelism, Graphs and Optimization. \\ This work was partially supported by CNPq and CAPES, Brazil.}
}

\author{Cl\'audia Linhares Sales \and
        Leonardo Sampaio \and
        Ana Silva
}


\institute{Departamento de Computa\c{c}\~ao, Universidade Federal do Cear\'a, Brazil. 
              \email{linhares@lia.ufc.br}           
           \and
           Centro de Ci\^encias, Universidade Estadual do Cear\'a, Brazil. 
            \email{leonardo.sampaio@uece.br}
            \and
            Departamento de Matem\'atica, Universidade Federal do Cear\'a, Brazil. 
            \email{ana.silva@mat.ufc.br}
}

\date{Received: date / Accepted: date}

\maketitle

\begin{abstract}
A b-coloring of the vertices of a graph is a proper coloring where each color class contains a vertex which is adjacent to each other color class. The b-chromatic number of $G$ is the maximum integer $\bi(G)$ for which $G$ has a b-coloring with $\bi(G)$ colors. A graph $G$ is b-continuous if $G$ has a b-coloring with $k$ colors, for every integer $k$ in the interval $[\chi(G),\bi(G)]$. It is known that not all graphs are b-continuous. Here, we investigate whether the lexicographic product $G[H]$ of b-continuous graphs $G$ and $H$ is also b-continuous. Using homomorphisms, we provide a new lower bound for $\bi(G[H])$, namely $\bi(G[K_t])$, where $t=\bi(H)$, and prove that if $G[K_\ell]$ is b-continuous for every positive integer $\ell$, then $G[H]$ admits a b-coloring with $k$ colors, for every $k$ in the interval $[\chi(G[H]),\bi(G[K_t])]$. We also prove that $G[K_\ell]$ is b-continuous, for every positive integer $\ell$, whenever $G$ a $P_4$-sparse graph, and we give further results on the b-spectrum of $G[K_\ell]$, when $G$ is chordal.
\end{abstract}
 
\section{Introduction}\label{intro}

Given a simple graph $G$\footnote{The graph terminology used in this paper follows \cite{BM08}.}, and a function $c: V(G)\rightarrow\{1, \cdots, k\}$, we say that $c$ is a \emph{proper coloring of $G$ with $k$ colors} if $c(u) \neq c(v)$ for every $uv \in E(G)$. A value $i\in \{1,\cdots,k\}$ is called \emph{color $i$}, while the subset $c^{-1}(i)$ is called \emph{color class $i$}.
Graph colorings are a very useful model for situations in which a set of objects is to be partitioned according to some prescribed rules.
For example, problems of \emph{scheduling}~\cite{Werra.85}, \emph{frequency assignment}~\cite{Gamst.86}, \emph{register allocation}~\cite{Chow.Hennessy.84,Chow.Hennessy.90}, and the \emph{finite element method}~\cite{Saad.96}, are naturally modelled by colourings.
In these applications, one is interested in finding a proper coloring with the smallest number of colors.
This motivates the definition of the \emph{chromatic number} of $G$, denoted $\chi(G)$, the smallest integer $k$ for which $G$ has a proper coloring with $k$ colors.
Deciding if a graph admits a proper colouring with $k$ colours is an NP-complete problem, even if $k$ is not part of the input~\cite{Hol81}.
The chromatic number is also hard to approximate: for all $\epsilon > 0$, there is no algorithm that approximates the chromatic number within a factor of $n^{1 - \epsilon}$ unless P = NP~\cite{Has96,Zuc07}.

One approach to obtain proper colorings of a graph is to use coloring heuristics.
Consider a proper coloring $c$ of graph $G$, for which we want to reduce the number of colors.
A vertex $v$ in color class $i$ is called a \emph{b-vertex of color $i$} if $v$ has at least one neighbor in color class $j$, for every $j\neq i$.
If color $i$ has no $b$-vertices, we may recolor each $v$ in color class $i$ with some color that does not appear in the neighborhood of $v$.
In this way, we eliminate color $i$, and obtain a new coloring for $G$ that uses $k - 1$ colors.
The procedure may be repeated until we reach a coloring in which every color contains a $b$-vertex.
Such a coloring is called a \emph{$b$-coloring}. Clearly, the described procedure cannot decrease the number of colors used in a proper coloring of $G$ with $\chi(G)$ colors. Therefore, we are actually interested in investigating the worst-case scenario for the described procedure. 
This motivates the definition of the \emph{$b$-chromatic number} of a graph $G$, denoted $\bi(G)$, being the largest $k$ such that $G$ has a $b$-coloring with $k$ colors.

This concept was introduced by Irving and Manlove in~\cite{IM99}, where they prove that  determining the $b$-chromatic number of a graph is an NP-complete problem. In fact, it remains so even when restricted to bipartite graphs~\cite{KTV02}, connected chordal graphs~\cite{HLS11}, and line graphs~\cite{CLMSSS15}.

In~\cite{KTV02} it is proved that $K'_{p, p}$, the graph obtained from $K_{p,p}$ by removing a perfect matching, admits $b$-colorings only with $2$ or $p$ colors.
And in~\cite{BCF07}, the authors prove that, for every finite $S\subset \Nat-\{1\}$, there exists a graph $G$ that admits a b-coloring with $k$ colors if and only if $k\in S$.
These facts motivate the definition of b-spectrum and of b-continuous graphs, introduced in~\cite{BCF03}.
The \emph{b-spectrum} of a graph $G$, denoted by $S_b(G)$, is the set containing every positive value $k$ for which $G$ admits a b-coloring with $k$ colors; and $G$ is said to be \emph{b-continuous} if $S_b(G)$ contains every integer in the closed interval $[\chi(G),\bi(G)]$. In the same article, they prove that interval graphs are $b$-continuous. This result was generalized for chordal graphs independently in~\cite{F04} and~\cite{KKV.04}.
Other examples of b-continuous graphs are cographs and $P_4$-sparse graphs~\cite{BDM+09}, as well as the more general class of $P_4$-tidy graphs~\cite{VBK10}; Kneser graphs $K(n,2)$ for $n\ge 17$~\cite{JaOm09}; regular graphs with girth at least~6 and not containing cycles of length~7~\cite{BK.12}; and, more recently, graphs with girth at least~10~\cite{SLS.16}. 

Given graphs $G=(V,E)$ and $H=(V,E)$, the \emph{lexicographic product } of $G$ by $H$ is the graph $G[H] = (V',E')$, where 
$V' = V(G) \times V(H)$ and $(x,y)(x',y')\in E'$ if and only if either $x=x'$ and $yy' \in E(H)$ or 
$xx' \in E(G)$.  Intuitively, $G[H]$ is the graph obtained from $G$ and $H$ by replacing each vertex of $G$ with a copy of $H$ and adding every possible edge between two copies of $H$ if  and only if the corresponding vertices of $G$ are adjacent. 
For every $x \in V(G)$, we denote the copy of $H$ related to $x$ in $G[H]$ by $x[H]$.
The $b$-chromatic number of the lexicographic product of graphs was studied in~\cite{JaPe12}.
In~\cite{LVS15}, which was co-authored by the authors, they considered the following question, which we continue to investigate in this paper.

\begin{question}\label{mainquestion}
Is it true that $G[H]$ is b-continuous whenever $G$ and $H$ are?
\end{question}

We mention that a similar question is answered in the negative for the cartesian product and strong product of graphs. For the cartesian product, it suffices to observe that the cube $Q_3$, which is known to be non-b-continuous, can also be viewed as the cartesian product of $C_4$ by $C_4$. As for the strong product, consider $G=K_2$ and $H=K_n$, $n>3$. Observe that $G \times H$ is isomorphic to $K'_{n,n}$, and as mentioned before, the $b$-spectrum of $K'_{n,n}$ is $\{2,n\}$.

In~\cite{JaPe12}, the authors show that $\bi(G[H])\ge \bi(G)\bi(H)$, while in~\cite{GS75} it is shown that $\chi(G[H]) = \chi(G[K_{\chi(H)}]) \le \chi(G)\chi(H)$. Therefore, if $G[H]$ is b-continuous, then the closed interval $[\chi(G)\chi(H),\bi(G)\bi(H)]$ must be contained in the b-spectrum of $G[H]$. In~\cite{LVS15}, the authors show that this is the case when $\bi(H)>\chi(H)$.

\begin{theorem}[\cite{LVS15}]\label{LVS15}
If $G$ and $H$ are b-continuous and $\bi(H) > \chi(H)$, then $[\chi(G)\chi(H),\bi(G)\bi(H)]\subseteq S_b(G[H])$.
\end{theorem}

In~Section~\ref{sec:bhom}, we apply the concept of b-homomorphism to show that $\bi(G[H])$ is in fact at least $\bi(G[K_t])$, where $t=\bi(H)$. This improves the lower bound $\bi(G)\bi(H)$ given in~\cite{JaPe12}. We also show that if $G[K_x]$ is b-continuous for every integer $x$, then $S_b(G[H])$ contains every integer in the closed interval $[\chi(G[H]), \bi(G[K_t])]$, which by what is said before, contains the interval $[\chi(G)\chi(H), \bi(G)\bi(H)]$. This shows that an important step to answer Question \ref{mainquestion} is to first answer the particular case below. Observe that this is complementary to Theorem~\ref{LVS15}.

\begin{question}\label{question:GKxBcontinuous}
Let $G$ be a b-continuous graph and $\ell$ be any positive integer. Is $G[K_\ell]$ b-continuous?
\end{question}

In~\cite{LVS15}, the authors have answered this question positively for chordal graphs, and in~Section~\ref{sec:P4sparse}, we show that this is also the case for $P_4$-sparse graphs. Below, we show that if the roles of $G$ and $K_\ell$ are reversed, then the answer is yes.


\begin{theorem}\label{teo:Kn[H]} If $H$ is a $b$-continuous graph, then $K_\ell[H]$ is $b$-continuous, for every positive integer $\ell$.
\end{theorem}
\begin{proof}
Denote $K_\ell$ by $G$. First observe that $\chi(G[H]) = |V(G)|(\chi(H))$. Now, take any $b$-coloring $c$ of $G[H]$ with $k$ colors, $k > \chi(G[H])$. Then, for some $v \in V(G)$, $v[H]$ is colored with $k' > \chi(H)$ colors.
Observe as well that the colors of $v[H]$ only occur in $v[H]$. Therefore, $c$ restricted to
$v[H]$ is a $b$-coloring with $k'$ colors, and since $H$ is $b$-continuous, it can be turned into a $b$-coloring with $k'-1$-colors, to produce a $b$-coloring of $G[H]$ with $k-1$ colors.
\end{proof}

In Section~\ref{sec:chordal}, we further investigate the b-spectrum of $G[H]$, when $G$ is chordal. We prove that if $G$ is chordal, $H$ is b-continuous and $k$ is an integer in the interval $[\chi(G[H]),\bi(G[H])]$ such that $G[H]$ does not have a b-coloring with $k$ colors, then $k$ is in the open interval $(\bi(G[K_t]), n_H\chi(G))$, where $t=\bi(H)$ and $n_H=\lvert V(H)\rvert$. If this interval is empty, it means that $G[H]$ is b-continuous. Therefore, a good question is the following:

\begin{question}
What are the graphs $G$ and $H$ such that $G$ is chordal, $H$ is b-continuous and $\bi(G[K_{\bi(H)}])\ge \lvert V(H)\rvert \chi(G)$?
\end{question}


\section{b-Homomorphism}\label{sec:bhom}

In this section, we investigate the concept of b-homomorphism. This will help us obtaining a new lower bound for $\bi(G[H])$, and showing that $S_b(G[H])$ contains the integers of the interval $[\chi(G[H]), \bi(G[K_{\bi(H)}])]$, whenever $H$ is b-continuous and $G[K_\ell]$ is b-continuous, for every $\ell$. We refer the reader to~\cite{HN.book} for an overview on graph homomorphisms.

Given any function $f:A\rightarrow B$ and a subset $A'\subseteq A$, we denote by $f(A')$ the set $\{f(a)\mid a\in A'\}$. Given graphs $G$ and $H$, and a function $f:V(G)\rightarrow V(H)$, we say that $f$ is a \emph{homomorphism from $G$ to $H$} if $f(u)f(v)\in E(H)$ for every $uv\in E(G)$; and that $f$ is a \emph{b-homomorphism from $G$ to $H$} if it is a homomorphism from $G$ to $H$ and, for every $x\in V(H)$, there exists $u\in f^{-1}(x)$ such that $f(N(u)) = N(x)$. If such a function exists, we write $G\bhom H$. Note that $f$ is surjective by definition.

\begin{proposition}\label{prop:transitivity}
If $F \bhom G$ and $G\bhom H$, then $F\bhom H$.
\end{proposition}
\begin{proof}
Let $f_1,f_2$ be homomorphisms from $F$ to $G$ and $G$ to $H$, respectively, and let $f = f_1\circ f_2$ (composition function of $f_1$ and $f_2$). It is known that $f$ is a homomorphism from $F$ to $H$, therefore we just need to prove that it is also a b-homomorphism. So, let $w\in V(H)$ be any vertex. Then, there exists $v\in f_2^{-1}(w)$ such that $f_2(N(v)) = N(w)$; in turn, there exists $u\in f_1^{-1}(v)$ such that $f_1(N(u)) = N(v)$. Therefore, $u\in V(F)$ is such that $f_2(f_1(N(u)) = f_2(N(v)) = N(w)$, as desired.
\end{proof}

\begin{lemma}\label{lem:lexHom}
If $F\bhom H$, then $G[F] \bhom G[H]$ and $F[G]\bhom H[G]$, $\forall G$.
\end{lemma}
\begin{proof}
Let $f$ be a b-homomorphism from $F$ to $H$ and define $f':V(G[F])\rightarrow V(G[H])$ as $f'(u,v) = (u, f(v))$. We prove that $f'$ is a b-homomorphism. First, note that if $(u,v)(u',v')$ is an edge in $G[F]$, then either $u\neq u'$ and $uu'\in E(G)$, in which case $(u,f(v))(u',f(v'))$ is an edge in $G[H]$, or $u = u'$ and $vv'\in E(F)$, in which case $f(v)f(v')\in E(H)$ and $(u,f(v))(u',f(v'))$ is an edge in $G[H]$. Therefore, $f'$ is a homomorphism.

Now, denote $G[F]$ by $F'$ and $G[H]$ by $H'$. Consider $(x,y)\in V(H')$, and let $v\in f^{-1}(y)$ be such that $f(N_F(v)) = N_H(y)$. We want to prove that $f'(N_{F'}(x,v)) = N_{H'}(x,y)$. Note that $f'(N_{F'}(x,v))\subseteq N_{H'}(x,y)$ since $f'$ is a homomorphism. Therefore, it remains to prove that $N_{H'}(x,y) \subseteq f'(N_{F'}(x,v))$, i.e., we need to prove that for each $(x',y')\in N_{H'}(x,y)$, there exists $(x'',v')\in N_{F'}(x,v)$ such that $f'(x'',v') = (x',y')$. 
So, consider any $(x',y') \in N_{H'}(x,y)$. Then, either $x\neq x'$ and $xx'\in E(G)$, or $x=x'$ and $yy'\in E(H)$. If the latter occurs, by the choice of $v$, there must exist $v'\in N_F(v)$ such that $f(v') = y'$; hence, $(x,v')$ is the desired vertex. If the former occurs, because $f$ is surjective, there must exist $v'\in V(F)$ such that $f(v') = y'$. Then, $(x',v')\in N_{F'}(x,v)$ and $f'(x',v') = (x',y')$, as desired.

Now we want to prove the second part of the lemma. So now let $F',H'$ denote $F[G],H[G]$, respectively, and consider $f'$ defined as $f'(u,v) = (f(u), v)$. If $(u,x)(v,y)\in F'$, then either $u=v$ and $xy\in E(G)$, or $u\neq v$ and $uv\in E(F)$. If the former occurs, then $f'(u,x)f'(u,y) = (f(u),x)(f(u),y)$ is also an edge in $H'$; and if the latter occurs, then $f(u)f(v)\in E(H)$ and again $f'(u,x)f'(v,y) = (f(u),x)(f(v),y)$ is also an edge in $H'$. Therefore, $f'$ is an homomorphism and it remains to show that it is a b-homomorphism.

So, consider any $(x,y)\in V(H')$, and let $v\in f^{-1}(x)$ be such that $f(N_F(v)) = N_H(x)$. We prove that $f'(N_{F'}(v,y)) = N_{H'}(x,y)$. For this, consider any $(u,w)\in N_{H'}(x,y)$. Then one of two cases occurs. If $u = x$ and $wy\in E(G)$, then we know that $(v,w) \in N_{F'}(v,y)$ and clearly $f'(v,w) = (x,w) = (u,w)$. So consider $u\neq x$ and $ux\in E(H)$. Then, by the choice of $v$, there must exist $u'\in N_F(v)$ such that $f(u') = u$. Therefore, $(u',w) \in N_{F'}(v,x)$ is such that $f'(u',w) = (u,w)$.
\end{proof}

It is easy to verify that the following holds.

\begin{proposition}\label{prop:bhomKm}
$H\bhom K_m$ if and only if $H$ has a b-coloring with $m$ colors.
\end{proposition}

The next lemma and corollary show the importance of Question \ref{question:GKxBcontinuous}.

\begin{lemma}\label{cor:relations}
Consider any $G, H$, and let $h = \chi(H)$, $g = \chi(G)$, $p = \bi(H)$ and $q=\bi(G)$. Then, the following hold:
\begin{enumerate}[1.]
  \item $\bi(G[H])\ge \bi(G[K_p]) \ge p\cdot q = \bi(K_q[H])$; 
  \item $\chi(G[H]) = \chi(G[K_h]) \le g\cdot h = \chi(K_g[H]) \le p\cdot q$; and
  \item $\bigcup_{x\in S_b(H)}S_b(G[K_x]) \subseteq S_b(G[H])$.
\end{enumerate}  
\end{lemma}
\begin{proof}
Recall that $\bi(F[J]) \ge \bi(F)\bi(J)$~\cite{JaPe12}; this explains the second inequality in~(1). Also, recall that $\chi(F[J]) = \chi(F[K_{\chi(J)}]) \le \chi(F)\chi(J)$~\cite{GS75}, which explains the first equality and first inequality in (2). 
Now, consider any positive integer $\ell$, and let $c$ be a proper coloring of $H' = K_\ell[H]$. Note that the colors used in any pair of distinct copies of $H$ in $H'$ are disjoint. Therefore we get $\chi(K_{\ell}[H]) = \ell\chi(H)$ and $\bi(K_\ell[H]) = \ell\bi(H)$. This explains the equality in (1) and the second equality in (2). Also, because $\chi(F)\le \bi(F)$ for every graph $F$, we get: $\chi(K_g[H]) \le \bi(K_g[H]) = g\bi(H) \le \bi(G)\bi(H) = p\cdot q$. This explains the last inequality in (2).

Now, we prove the first inequality in (1). Let $p' = \bi(G[K_p])$. By Proposition \ref{prop:bhomKm}, we get $H\bhom K_p$, and $G[K_p]\bhom K_{p'}$; by Lemma \ref{lem:lexHom}, we get $G[H]\bhom G[K_p]$; and by Proposition \ref{prop:transitivity}, we get $G[H] \bhom K_{p'}$. Therefore, by Proposition \ref{prop:bhomKm}, we get $\bi(G[H]) \ge p'$. 

It remains to prove (3). Let $x\in S_b(H)$ and $y\in S_b(G[K_x])$. By Proposition \ref{prop:bhomKm}, $H\bhom K_x$ and $G[K_x] \bhom K_y$, and by Lemma \ref{lem:lexHom} and Proposition~\ref{prop:transitivity}, $G[H] \bhom K_y$. Hence, $y\in S_b(G[H])$.
\end{proof}

The next corollary easily follows from Lemma  \ref{cor:relations}.

\begin{corollary}\label{cor:spectrum}
If $G[K_x]$ is b-continuous, for every positive integer $x$, and $H$ is b-continuous, then (below, $t$ denotes $\bi(H)$)
\[[\chi(G[H]),\bi(G[K_t])]\subseteq S_b(G[H]).\]
\end{corollary}


Given a graph $H$, let $x$ denote the value $2\lvert V(H)\rvert +\Delta(H) +1$. In~\cite{JaPe12}, the authors prove that if $n\ge 2x+1$, then $\bi(P_n[H]) = x$. Therefore, for any $k\ge 5$ and $n\ge 4k+7$, we have that $\bi(P_n[P_k]) = 2k+3$, while $\bi(K_3[P_k]) = \bi(P_n[K_3]) = 9$. This tells us that $\bi(G[H])$ can be arbitrarily larger than $\bi(G[K_{\bi(H)}])$.  Nevertheless, we give an example where $\bi(G[K_\ell])$ is strictly larger than $\ell\bi(G)$, showing that our lower bound improves the best previously known lower bound for $\bi(G[H])$, namely $\bi(G)\bi(H)$~\cite{JaPe12}. For this, consider the tree $T$ obtained from the $P_5$, $(v_1,v_2,x,v_3,v_4)$, by adding one pendant leaf at $v_2$, one at $v_3$, two at $v_1$ and two at $v_4$. It is not hard to verify that $\bi(T) = 3$. On the other hand, one can verify that the precoloring $\{(3,4), (1,2), (3,6), (4,7), (5,6)\}$ of $P_5[K_2]$ can be completed into a b-coloring of $T[K_2]$ with 7 colors, thus showing that $\bi(T[K_2]) > 2\bi(T)$. 

Corollary~\ref{cor:spectrum} shows the importance of knowing the value $\bi(G[K_t])$. In~\cite{IM99}, the authors introduce an upper bound for $\bi(G)$. Observe that if $G$ has a b-coloring with $k$ colors, then $G$ has at least $k$ vertices with degree at least $k-1$, namely the b-vertices. Therefore, if $m(G)$ is the largest $k$ for which $G$ has at least $k$ vertices with degree at least $k-1$, then $\bi(G)\le m(G)$. Observe that $m(G)$ can be easily computed by ordering the vertices of $G$ according to their degrees in a non-increasing way. As a consequence of the following proposition, we get that the distance $\bi(G[K_\ell]) - \ell\bi(G)$ is at most $\ell(m(G) - \bi(G))$. 

%
%

\begin{proposition}
Let $G$ be any graph and let $\ell$ be any positive integer. Then, $$m(G[K_\ell])= \ell m(G).$$ 
\end{proposition}
\begin{proof}
Denote $G[K_\ell]$ by $G'$ and $m(G)$ by $m$. First, we prove that there are at least $\ell m$ vertices of degree at least $\ell m-1$ in $G'$. For this, let $D$ be the subset of vertices of $G$ with degree at least $m-1$. By the definition of $m(G)$, there are at least $m$ such vertices. Also, for each $u\in D$ and each $v\in V(K_\ell)$, we have $d(u,v) = \ell d(u) + \ell - 1 \ge (m-1)\ell + \ell - 1 = \ell m - 1$. Therefore, the set $\bigcup_{u\in D}\{(u,v)\mid v\in V(K_\ell)\}$ contains the desired vertices.

Now, we prove that there are at most $\ell m$ vertices of degree at least $\ell m$, which implies that $m(G')$ cannot exceed $\ell m$. For this, just consider any $(u,v)\in V(G')$ such that $d_{G'}(u,v)\ge \ell m$. Since $d_{G'}(u,v) = \ell d(u) + \ell -1$, we get that $d(u) \ge \frac{1}{\ell}\left( \ell m - \ell +1 \right) \ge m - 1 + \frac{1}{\ell} > m-1$. Therefore, each such vertex in $G'$ defines a vertex of $G$ with degree at least $m$. By the definition of $m(G)$, one can see that there are at most $m$ such vertices in $G$, which implies that there are at most $\ell m$ vertices of degree at least $\ell m$ in $G'$.
\end{proof}


%

In the following sections, we investigate the answer to Question \ref{question:GKxBcontinuous} restricted to some known b-continuous classes.


\section{$P_4$-sparse graphs} \label{sec:P4sparse}

In this section, we prove the following theorem.

\begin{theorem}\label{thm:P4sparse}
Let $G$ be a $P_4$-sparse graph and $\ell$ be any positive integer. Then, $G[K_\ell]$ is b-continuous.
\end{theorem}

We mention that in~\cite{BDM+09}, the authors prove that $P_4$-sparse are b-continuous. Our proof generalizes theirs, since it also holds when $\ell = 1$. It is also worth mentioning that $P_4$-tidy graphs, a superclass of $P_4$-sparse graphs, are also b-continuous~\cite{VBK10}. A good question is whether our result can be generalized to $P_4$-tidy graphs.

Before we proceed, we need some definitions.
Consider a graph $G$; we say that $G$ is \emph{complete} if $E(G)$ contains every possible edge, and that $G$ is \emph{empty} if $E(G)$ is empty. Let $X\subseteq V(G)$; the subset $X$ is called a \emph{clique} if $G[X]$ is complete, and it is called a \emph{stable set} if $G[X]$ is empty. A \emph{matching} in $G$ is a collection of pairwise non-adjacent edges, while an \emph{antimatching} in $G$ is a matching in $\overline{G}$ (complement graph of $G$). Given disjoint subsets of vertices $U,W\subseteq V(G)$, we say that $U$ is \emph{complete to $W$} if every possible edge between $U$ and $W$ exists in $G$, and that $U$ is \emph{anti-complete to $W$} if $U$ is complete to $W$ in $\overline{G}$. 

Given disjoint graphs $G_1$ and $G_2$, the \emph{union of $G_1$ and $G_2$} is the graph $(V(G_1)\cup V(G_2), E(G_1)\cup E(G_2))$, while the \emph{join of $G_1$ and $G_2$} is obtained from their union by adding every possible edge between $G_1$ and $G_2$. 
Finally, let $C$ and $S$ be the complete and empty graphs on $n$ vertices, respectively, $C\cap S=\emptyset$, and add either a matching or an anti-matching between $C$ and $S$. Also, let $R$ be any subset disjoint from both $K$ and $S$. The \emph{spider operation} applied to $(C,S)$ and $R$ is obtained by adding every possible edge between the sets $R$ and $C$. The obtained graph is called a \emph{spider} and we say that $R$ is the \emph{head} of the spider. If $R=\emptyset$, we say that \emph{$(C,S)$ is a spider with empty head}. The following decomposition theorem is an important tool in our proof.

\begin{theorem}[\cite{Ho95,JaOl95}]
\label{thm:decp4}
If $G$ is a non-trivial $P_4$-sparse graph, then exactly one of the following holds:
\begin{enumerate}
\item $G$ is the union of two $P_4$-sparse graphs; or
\item $G$ is the join of two $P_4$-sparse graphs; or
\item $G$ is a spider whose head is either empty or a $P_4$-sparse graph.
\end{enumerate}
\end{theorem}

Given a $P_4$-sparse graph $G$, let $(T, t, {\cal X})$ be a tuple, where $T$ is a rooted tree having $I$ as internal vertices, $t$ is a function $t: I \rightarrow \{union,\ join,\ spider\}$, and ${\cal X} = \{X_j\}_{j\in V(T)\setminus I}$ associates to each leaf $i$ of $T$ a subset $X_i\subseteq V(G)$, disjoint from $X_j$ for every $j\neq i$, and such that $X_i$ is either a clique, or a stable set, or induces a spider with empty head. We say that $D = (T, t, {\cal X})$ is a \emph{primeval decomposition of $G$} if $G$ can be constructed from ${\cal X}$ by searching the tree in an upward way and applying the operation defined by the label on the respective internal node of $T$. We suppose that $D$ has a minimal number of internal vertices, i.e., if $x\in V(T)$ is such that $t(x) = union$, and $x$ is adjacent to leaves $y,z$, then at least one between $X_y$ and $X_z$ is not empty as otherwise the decomposition obtained from $D$ by removing $y,z$ and relating $x$ to $X_y\cup X_z$ is also a primeval decomposition of $G$. A similar argument can be done when $t(x) = join$ and $X_y,X_z$ are complete. For each node $i\in T$, we denote by $D^i = (T^i, t^i, {\cal X}^i)$ the tuple $D$ restricted to the subtree of $T$ rooted at $i$. Observe that $D^i$ itself is a primeval decomposition of $G^i$, the subgraph formed by ${\cal X}^i$. We mention that if $i$ is an internal node or $X_i$ is not a spider, then $V(G^i)$ is a module in $G$, i.e., for every $v\in V(G)\setminus V(G^i)$, we get that $v$ is either complete or anti-complete to $V(G^i)$.

In our proof, we start with a b-coloring $\psi$ of $G[K_\ell]$ and iteratively try to remove a fixed color, namely color~1, from some $X_i[K_\ell]$, where $i$ is a leaf in the primeval decomposition of $G$. While doing this, we allow for the b-vertices to lose color~1 in their neighborhoods. Therefore, if at the end no more vertex is colored with~1, then the obtained coloring is a b-coloring of $G[K_\ell]$ with $k-1$ colors. However, removing color~1 from a leaf is not always possible. When this happens, we restrict our attention to a subgraph $G'$ of $G$ such that $\psi$ restricted to $G'[K_\ell]$ is also a b-coloring with $k$ colors. If eventually we arrive at a clique and color~1 cannot be removed, we get that $\chi(G[K_\ell])\ge \omega(G[K_\ell]) = \lvert G'[K_\ell]\rvert=k$; hence no b-coloring with $k-1$ colors can exist and we are done. Before we proceed to our proof, we need some further definitions that tell us what is an acceptable coloring and how we can reduce the subgraph being investigated.

Consider any graph $H$ and let $c$ be a proper coloring of $V(H)$. We say that $u\in V(H)$ is a $b^*$-vertex in $c$ if $c(u)\neq 1$ and $u$ is adjacent to a vertex colored with $i$, for every color $i$ distinct from $1$ and $c(u)$; and we say that $c$ is a \emph{miss-1-b-coloring} of $H$ if there exists a $b^*$-vertex colored with $i$, for every color $i$ distinct from~1. Note that $b^*$-vertices are not necessarily non-adjacent to~$1$; therefore any b-vertex of $c$ satisfies the condition, i.e., a b-coloring with $k$ colors is also a miss-1-b-coloring. In addition, note that a miss-1-b-coloring where no vertex is colored with color 1 is a b-coloring with $k-1$ colors. In our proof, starting with a b-coloring with $k$ colors, we manipulate miss-1-b-colorings until color 1 disappears, or until we can prove that $k=\omega(\Gl)$. 

Now, consider a minimal primeval decomposition $D = (T, t, {\cal X})$ of $G$. Given a leaf $i$ of $T$, let $p_i$ be the parent node of $i$ in $T$. We know that $p_i$ is an internal node labeled with one of the operations in $\{union, join, spider\}$, which are binary operations. This means that $p_i$ has exactly one other child different from $i$, say $j$, and that, when constructing $G$, the subgraph $G[X_i]$ is operated with the subgraph $G^j$ by applying the operation $t(p_i)$. We introduce some reduction operations on $G$ that allow us to restricted our attention to a subgraph $G^*$ of $G$. Below, we show that $D$ can be easily adapted to a minimal primeval decomposition of $G^*$. We implicitly use these decompositions in the proof.
\begin{itemize}
\item If $X_i$ is a clique and $t(p_i)\neq join$, we say that $G^*$ is a \emph{c-reduction of $G$} if it is obtained from $G$ by removing every vertex of $G^j$ non-adjacent to $X_i$. A primeval decomposition $D^*$ of $G^*$ can be obtained from $D$ by removing $i$ and the subtree rooted in $j$ from $T$, and relating $p_i$ either with $X_i$ when $t(p_i) = union$, or with $X_i\cup C$ when $t(p_i) = spider$, where $G^j$ is a spider with empty head and clique set $C$;

\item If $X_i$ is a stable set, we say that $G^*$ is an \emph{s-reduction of $G$} if it is obtained from $G$ by removing every vertex of $X_i$ except one, say $\{v\}$. A primeval decomposition $D^*$ of $G^*$ can be obtained from $D$ by relating $i$ to $\{v\}$. Denote node $i$ by $p$ and note that it is a clique leaf node in the  new decomposition;

\item If $X_i$ is a spider with partition $(C,S)$, we say that $G^* = G - S$ is a \emph{p-reduction of $G$}.  A primeval decomposition $D^*$ of $G^*$ can be obtained from $D$ by changing the label of $p_i$ to $join$.
\end{itemize}

Note that the decomposition $D^*$ may not be a minimal decomposition anymore, but that it is not hard to obtain a minimal decomposition from $D^*$ simply by removing $X_j,X_h$ that have a common parent $p$ and relating $p$ to $X_j\cup X_h$, whenever either $X_j,X_h$ are stable sets and $t(p) = union$, or $X_j,X_h$ are cliques and $t(p)=join$. Therefore, at all times we consider the decomposition being used to be a minimal decomposition.
We also want the reader to remark that the c-reduction and p-reduction are not the same. The main difference is that, in the c-reduction, vertices from $G^j$ are removed, while in the p-reduction, it is  $X_i$ that loses vertices.
If $G^*$ is a c-reduction, or an s-reduction, or a p-reduction of $G$, we say that it is a \emph{reduction of $G$}. Observe that if $G^*$ is a reduction of $G$, then $G^*$ is an induced subgraph of $G$. In the proof, we iteratively reduce the graph. For this, if $G_i$ is the current subgraph being considered, we implicitly deal with a minimal primeval decomposition $D=(T,t,{\cal X})$ of $G_i$ and with a b-coloring $\psi_i$ of $G_i[K_\ell]$, and we say that  a leaf node $j\in V(T)$ \emph{contains color~1} if $1\in \psi_i(X_j[K_\ell])$.

\begin{proof}[of Theorem \ref{thm:P4sparse}]
We construct a sequence of pairs $(G_i,\psi_i)_{i\ge 0}^p$ such that:
\begin{enumerate}
  \item[(i)] $\psi_i$ is a miss-1-b-coloring of $\Gil{i}$ with $k$ colors, for every $i\in\{0,\cdots,p\}$;
  \item[(ii)] For each $i\in \{1,\cdots, p\}$, either $G_i$ is a reduction of $G_{i-1}$, or $G_i = G_{i-1}$ and in $\psi_i$ there is either one less leaf or one less clique leaf containing color 1; and
  \item[(iii)] Either color class 1 is empty in $\psi_p$, or $\omega(G_p[K_\ell]) = k$.
\end{enumerate}

For the reader to better understand Condition~(ii), we mention that sometimes we are able to decrease the number of clique leaves containing color~1, but without decreasing the number of leaves containing color~1. This is because a recoloring is made in a way that color~1 appears in a non-clique leaf that did not contain color~1 before. 

Naturally, we start by setting $(G_0,\psi_0)$ to $(G,\psi)$. We show how to construct the sequence and prove that if $\psi_p$ is a b-coloring of $\Gil{p}$ with $k-1$ colors (which is the case if $\omega(G_p[K_\ell]) < k$), then a b-coloring of $\Gl$ with $k-1$ colors can be obtained. In fact, throughout the construction, whenever $G_{i+1}$ is a reduction of $G_i$, we explain how to obtain a b-coloring of $\Gil{i}$ with $k-1$ colors, given a b-coloring of $\Gil{i+1}$ with $k-1$ colors. This gives us what we need. Because $\Gil{p}\subseteq G[K_\ell]$, we get that if $\omega(\Gil{p}) = k$, then $\chi(\Gl) \ge \omega(\Gl)\ge \omega(\Gil{p}) = k$ and we are done (no b-colorings of $\Gl$ with $k-1$ colors can exist). 

Consider we are at step $i$ of the construction and let $f$ be any leaf of $G_i$ containing color $1$ (if no such leaf exists, we are done since there are no more vertices colored with $1$). If $X_f = V(G_i)$, by the definition of miss-1-b-coloring we get that each color class distinct from $1$ is also non-empty; hence either $\Gil{i}$ is the complete graph with $k$ vertices, in which case we are done, or $G_i$ is a spider with empty head. So suppose the latter occurs and let $(C,S)$ be the partition of $V(G_i)$ such that $C$ is a clique and $S$ is a stable set of $G_i$. Note that $k\ge \ell\lvert C\rvert = \omega(\Gil{i})$, and that $d(u) \le \ell(\lvert C\rvert - 1)+\ell-1 = \ell\lvert C\rvert -1$, for every $u\in S[K_\ell]$. Therefore, if $k> \ell\lvert C\rvert +1$, then no vertex of $S[K_\ell]$ can be a $b^*$-vertex; but then we have at most $\ell\lvert C\rvert < k-1$ $b^*$-vertices, namely the vertices in $C[K_\ell]$, contradicting the fact that $\psi_i$ is a miss-1-b-coloring of $\Gil{i}$. Hence, either $k = \omega(\Gil{i})$ and we are done, or $k=\omega(\Gil{i})+1$ and a b-coloring with $k-1$ colors of $\Gil{i}$ can be easily obtained.


 Therefore, suppose that $X_f\neq V(G_i)$, and let $\odot$ be the label of the parent node of $f$ in $T$. Also, let $G'$ be the subgraph operated with $X_f$ in the construction of $G_i$, and denote by $H$ the graph $X_f\odot G'$. 

 First, suppose that $X_f$ is a stable set, and let $u\in V(X_f)$ be such that $1\in \psi_i(u[K_\ell])$. We suppose that $\lvert V(X_f)\rvert >1$ as otherwise we can treat $X_f$ as a clique. If there exists $c\in \psi_i(X_f[K_\ell])\setminus \psi_i(u[K_\ell])$, we switch colors $1$ and $c$ in $u[K_\ell]$. We can repeat this argument for other vertices of $X_f$ containing color~1. Therefore, we can suppose that $\psi_i(X_f[K_\ell])\subseteq \psi_i(u[K_\ell])$. In this case, we obtain an s-reduction $G_{i+1}$ of $G_i$ by removing every vertex of $X_f\setminus\{u\}$, and we let $\psi_{i+1}$ be $\psi_i$  restricted to $\Gil{i+1}$. Clearly (i) and (ii) hold. Also, if $\gamma$ is a b-coloring of $\Gil{i+1}$ with $k-1$ colors, then by coloring $v[K_\ell]$ with $\gamma(u[K_\ell])$, for each $v\in V(X_f)\setminus\{u\}$, we obtain a b-coloring of $\Gil{i}$ with $k-1$ colors.

Now, suppose that $X_f$ is a spider with empty head. Let $(C,S)$ be the partition of $V(X_f)$ where $C$ is a clique and $S$ is a stable set; denote by $C',S'$ the subsets $C[K_\ell],S[K_\ell]$, respectively. Observe that $\psi_i(G'[K_\ell])\cap \psi_i(C') = \emptyset$. If $1\notin \psi_i(C')$, change the color of every  $u\in S'$ colored with $1$ to $c\in \psi_i(C'\setminus N_{C'}(u))$; such a color exists since $u$ is not complete to $C'$. If $1\in \psi_i(C')$ and there exists $c\in \psi_i(S')\setminus \psi_i(C'\cup G'[K_\ell])$, then switch colors $1$ and $c$ in $X_f[K_\ell]$ and proceed as before. Finally, if $\psi_i(S')\subseteq \psi_i(C'\cup V(G'[K_\ell]))$, then let $G_{i+1}$ be a p-reduction of $G_i$ obtained by removing $S$, and let $\psi_{i+1}$ be equal to $\psi_i$ restricted to $\Gil{i+1}$. One can verify that (i) and (ii) hold. Now, consider $u$ to be any vertex of $G'$ (recall that $V(G')\neq \emptyset$ since $X_f\neq V(G_i)$). If $\gamma$ is a b-coloring of $\Gil{i+1}$ with $k-1$ colors, then a b-coloring of $\Gil{i}$ with $k-1$ colors can be obtained by giving colors $\gamma(u[K_\ell])$ to $v[K_\ell]$ for every $v\in S$.

Now, suppose that $X_f$ is a clique. If $\odot = join$, let $c\in \psi_i(G'[K_\ell])$ and let $\psi_{i+1}$ be obtained by switching colors 1 and $c$ in $H[K_\ell]$. Since $V(H$) is a module and $u$ is complete to $V(H[K_\ell])\setminus\{u\}$ for every $u\in X_f[K_\ell]$, we know that color $c$ cannot lose all of its $b^*$-vertices, and that any $t\in V(\Gil{i})\setminus V(H)$ is adjacent to the same set of colors (i.e., (i) holds). Also, because $D$ is minimal, we know that $G'$ is not a clique; hence $\psi_{i+1}$ has one less clique leaf containing color~1, i.e., (ii) holds for $\psi_{i+1}$. We can therefore suppose that $\odot\neq join$.
Let $N$ denote $N_{G'}(X_f)$, and first suppose that there exists a color $c\neq 1$ that appears in $G'[K_\ell]$ but does not appear in $X_f[K_\ell]\cup N[K_\ell]$. Note that color $c$ cannot appear in $N(H[K_\ell])$ because $V(H)$ is a module in $G_i$. If $\psi_{i+1}$ is obtained by switching colors~$1$ and~$c$ in $X_f[K_\ell]$, then (i) and (ii) hold. 

Finally, suppose that $X_f$ is a clique, $\odot\neq join$, and that $\psi_i(G'[K_\ell])\subseteq \psi_i(X_f[K_\ell]\cup N[K_\ell])$. Note that $\odot\neq join$ implies that $V(G')\setminus N\neq \emptyset$. Let $G_{i+1}$ be a c-reduction of $G_i$ obtained by removing $V(G')\setminus N$, and let $\psi_{i+1}$ be equal to $\psi_i$ restricted to $\Gil{i+1}$. Note that the colors in $V(G'[K_\ell])\setminus N[K_\ell]$ are redundant in $N(t)$, for every $t\in N(H[K_\ell])$. Therefore, we get that $\psi_{i+1}$ is a miss-1-b-coloring of $\Gil{i+1}$. It remains to show that if $\gamma$ is a b-coloring of $\Gil{i+1}$ with $k-1$ colors, then a b-coloring of $\Gil{i}$ with $k-1$ colors can be obtained. If $\odot$ is a spider operation, note that $G'[K_\ell]\setminus N[K_\ell]$ is the union of cliques of size $\ell$. Consider any $u\in V(X_f)$, and give colors $\gamma(u[K_\ell])$ to each such clique. So suppose that $\odot = union$. Because $X_f[K_\ell]$ is a clique and $\psi_i(G'[K_\ell])\subseteq \psi_i(X_f[K_\ell]) = \ell\lvert X_f\rvert$, we get that $\chi(G'[K_\ell])\le \ell\lvert X_f\rvert$ and we can simply optimally color $G'[K_\ell]$ with the colors in $X_f[K_\ell]$.

Because there is a finite number of vertices colored with 1, the previous process eventually stops, i.e., condition (iii) holds for some $p\ge 1$. 
\end{proof}


\section{Chordal graphs}\label{sec:chordal}

It has been proved that chordal graphs are $b$-continuous~\cite{F04}. Here, we investigate some aspects regarding the lexicographic product involving chordal graphs.
We first recall the following result, which gives us the corollary below by applying Corollary~\ref{cor:spectrum}.

\begin{theorem}[\cite{LVS15}]
Let $G$ be a chordal graph and $n$ be any positive integer. Then $G[K_n]$ is chordal. 
\end{theorem}

\begin{corollary}\label{cor:chordal}
Let $G$ be a chordal graph and $H$ be any b-continuous graph. Then $[\chi(G[H]), \bi(G[K_t])]\subseteq S_b(G[H])$, where $t = \bi(H)$.
\end{corollary}

In the next lemma, we further increase the known b-spectrum of $G[H]$, when $G$ is chordal. A \emph{simplicial vertex} is a vertex whose neighborhood is a clique. An order $(v_1,\cdots,v_n)$ of the vertices of a graph $G$ is called a \emph{perfect elimination order} if $v_i$ is a simplicial vertex in $G[\{v_i,\cdots,v_n\}]$, for every $i\in \{1,\cdots,n\}$. It is well known that $G$ is chordal if and only if $G$ admits a perfect elimination order~\cite{FG.65}. Below, $\omega(G)$ denotes the maximum size of a clique in~$G$.

\begin{lemma}
Let $G$ be a chordal graph and $H$ be any graph. If $\psi$ is a b-coloring of $G[H]$ with $k$ colors and $k>n_H\omega(G)$, where $n_H=\lvert V(H)\rvert$, then there exists a b-coloring of $G[H]$ with $k-1$ colors.
\end{lemma}
\begin{proof}
Let $(v_1,\cdots,v_n)$ be a perfect elimination order of $G$, and for each $i$ denote by $G_i$ the subgraph $G[\{v_i,\cdots,v_n\}]$, and by $\psi_i$ the coloring $\psi$ restricted to $G_i[H]$. Also, for each color $c$ and index $i$, denote by $B_{i,c}$ the set of b-vertices of color $c$ in $\psi_i$. Now, let $i$ be minimum such that $\psi_{i+1}$ is not a b-coloring of $G_{i+1}[H]$ with $k$ colors. This means that there exists a color $c$ such that $B_{i,c}\subseteq N_{G_i[H]}(v_i[H])$, and a color $c'\in \psi_i(v_i[H])$ such that every neighbor of $(v_j,u)$ colored with $c'$ is contained in $v_i[H]$, for every $(v_j,u)\in B_{i,c}$. Also, because $v_i$ is simplicial in $G_i$, we get that in fact  there exists $v_j\in N_{G_i}(v_i)$ such that $B_{i,c}\subseteq v_j[H]$. Finally, since $k>n_H\omega(G)$ and $N_{G_i}(v_i)$ is a clique, there exists a color $c''$ that does not appear in $v_i[H]\cup N_{G_i[H]}(v_i[H])$. 

Now, switch colors $c$ and $c'$ in $v_j[H]$, and colors $c'$ and $c''$ in $v_i[H]$. Because $B_{i,c}\subseteq v_j[H]$, we know that there are no remaining b-vertices in color class $c$; so let $\psi'$ be obtained by changing the color of each $x$ colored with $c$ to any color that does not appear in its neighborhood. By the choice of colors, one can verify that $\psi'$ is a b-coloring of $G_i[H]$ with $k-1$ colors. Finally, for $\ell$ equal to $i-1$ down to $1$, because $N_{G_\ell}(v_\ell)$ is a clique, we know that at most $n_H(\omega(G)-1)$ colors appear in $N_{G_\ell[H]}(v_\ell[H])$. Therefore, since $k>n_H\omega(G)$, there are at least $n_H$ colors with which we can color $v_\ell[H]$.
\end{proof}

\begin{corollary}
Let $G$ be a chordal graph and $H$ be a graph with $n_H$ vertices. If $\bi(G[H])\ge n_H\chi(G)$, then $[n_H\chi(G),\bi(G[H])]\subseteq S_b(G[H])$.
\end{corollary}

\vspace{-0.4cm}
\bibliography{bibliografia}

\begin{thebibliography}{10}

\bibitem{BK.12}
R.~Balakrishnan and T.~Kavaskar.
\newblock b-coloring of kneser graphs.
\newblock {\em Discrete Appl. Math.}, 160:9--14, 2012.

\bibitem{BCF03}
D.~Barth, J.~Cohen, and T.~Faik.
\newblock Complexity of determining the b-continuity property of graphs.
\newblock Technical Report PRiSM Technical Report 2003/37, Universit\'e de
  Versailles, 2003.

\bibitem{BCF07}
D.~Barth, J.~Cohen, and T.~Faik.
\newblock On the b-continuity property of graphs.
\newblock {\em Discrete Applied Mathematics}, 155(13):1761--1768, 2007.

\bibitem{BM08}
A.~Bondy and U.S.R. Murty.
\newblock {\em Graph Theory}.
\newblock Springer, 2008.

\bibitem{BDM+09}
F.~Bonomo, G.~Dur\'an, F.~Maffray, J.~Marenco, and M.~Valencia-Pabon.
\newblock On the $b$-coloring of cographs and $p_4$-sparse graphs.
\newblock {\em Graphs and Combinatorics}, 25:153--167, 2009.

\bibitem{CLMSSS15}
V.~Campos, C.~Lima, N.A. Martins, L.~Sampaio, M.C. Santos, and A.~Silva.
\newblock The b-chromatic index of graphs.
\newblock {\em Discrete Mathematics}, 338:2072--2079, 2015.

\bibitem{Chow.Hennessy.84}
F.~Chow and J.~Hennessy.
\newblock Register allocation by priority-based coloring.
\newblock {\em ACM SIGPLAN Notices}, 19:222--232, 1984.

\bibitem{Chow.Hennessy.90}
F.~Chow and J.~Hennessy.
\newblock The priority-based coloring approach to register allocation.
\newblock {\em ACM Transactions on Programming Languages and Systems},
  12:501--536, 1990.

\bibitem{F04}
T.~Faik.
\newblock About the b-continuity of graphs.
\newblock In {\em Workshop on Graphs and Combinatorial Optimization},
  volume~17, pages 151--156, 2004.

\bibitem{FG.65}
D.R. Fulkerson and O.~Gross.
\newblock Incidence matrices and interval graphs.
\newblock {\em Pacific Journal of Mathematics}, 15:835--855, 1965.

\bibitem{Gamst.86}
A.~Gamst.
\newblock Some lower bounds for the class of frequency assignment problems.
\newblock {\em IEEE Transactions on Vehicular Technology}, 35(8--14), 1986.

\bibitem{GS75}
D.~Geller and S.~Stahl.
\newblock The chromatic number and other functions of the lexicographic
  product.
\newblock {\em Journal of Combinatorial Theory B}, 19:87--95, 1975.

\bibitem{HLS11}
F.~Havet, C.~Linhares Sales, and L.~Sampaio.
\newblock $b$-coloring of tight graphs.
\newblock {\em Discrete Applied Mathematics}, 160(18):2709 -- 2715, 2012.

\bibitem{HN.book}
Pavol Hell and Jaroslav Nesetril.
\newblock {\em Graphs and Homomorphisms}.
\newblock Oxford University Press, 2004.

\bibitem{Ho95}
C.~Ho\`ang.
\newblock {\em Perfect graphs}.
\newblock PhD thesis, School of Computer Science, McGill University, Montreal,
  1995.

\bibitem{Hol81}
I.~Holyer.
\newblock The {NP}-completeness of edge-coloring.
\newblock {\em SIAM Journal on Computing}, 10(4):718--720, 1981.

\bibitem{Has96}
J.~H\r{a}stad.
\newblock Clique is hard to approximate within $n^{1 - \epsilon}$.
\newblock In {\em Acta Mathematica}, pages 627--636, 1996.

\bibitem{IM99}
R.~W. Irving and D.~F. Manlove.
\newblock The b-chromatic number of a graph.
\newblock {\em Discrete Applied Mathematics}, 91(1-3):127--141, 1999.

\bibitem{JaPe12}
M.~Jakovac and I.~Peterin.
\newblock On the b-chromatic number of some graph products.
\newblock {\em Studia Scientiarum Mathematicarum Hungarica}, 49(2):156--169,
  2012.

\bibitem{JaOl95}
Beverly Jamison and Stephan Olariu.
\newblock Linear time optimization algorithms for p4-sparse graphs.
\newblock {\em Discrete Applied Mathematics}, 61(2):155 -- 175, 1995.

\bibitem{JaOm09}
R.~Javadi and B.~Omoomi.
\newblock On b-coloring of the kneser graphs.
\newblock {\em Discrete Mathematics}, 309(13):4399--4408, 2009.

\bibitem{KKV.04}
J.~K\'ara, J.~Kratochv\'il, and M.~Voigt.
\newblock b-continuity.
\newblock Technical report, Technical University Ilmenau, Faculty of
  Mathematics and Natural Sciences, 2004.

\bibitem{KTV02}
J.~Kratochv\'{\i}l, Z.~Tuza, and M.~Voigt.
\newblock On the $b$-chromatic number of graphs.
\newblock In {\em Graph-Theoretic Concepts in Computer Science}, volume 2573 of
  {\em Lecture Notes in Computer Science}, pages 310--320. Springer Berlin /
  Heidelberg, 2002.

\bibitem{Saad.96}
Y.~Saad.
\newblock {\em Iterative Methods for Sparse Linear Systems}.
\newblock PWS Publishing Company, Boston, MA, USA, 1996.

\bibitem{LVS15}
C.~Linhares Sales, R.~Vargas, and L.~Sampaio.
\newblock b-continuity and the lexicographic product of graphs.
\newblock {\em Electronic Notes in Discrete Mathematics}, 2015.
\newblock LAGOS'15 – V Latin-American Algorithms, Graphs and Optimization
  Symposium, to appear.

\bibitem{SLS.16}
A.~Silva and C.~Linhares-Sales.
\newblock Graphs with large girth are b-continuous.
\newblock http://arxiv.org/abs/1602.01298, 2016.

\bibitem{VBK10}
C.~I.~B. Velasquez, F.~Bonomo, and I.~Koch.
\newblock On the b-coloring of p4-tidy graphs.
\newblock {\em Discrete Applied Mathematics}, 159(1):60--68, 2011.

\bibitem{Werra.85}
D.~Werra.
\newblock An introduction to timetabling.
\newblock {\em European Journal of Operations Research}, 19:151--161, 1985.

\bibitem{Zuc07}
David Zuckerman.
\newblock Linear degree extractors and the inapproximability of max clique and
  chromatic number.
\newblock {\em Theory of Computing}, 3(6), 2007.

\end{thebibliography}
\bibliographystyle{plain}

\end{document}